\documentclass[11pt]{article}

\usepackage{lscape}

\usepackage{amsmath}
\usepackage{epsfig,cancel,amsthm,amssymb}
\usepackage{color, soul}
\usepackage{wrapfig}
\usepackage{graphicx,framed}
\def \eqref#1{(\ref{#1})}

\makeatletter
\@addtoreset{equation}{section}
\makeatother

\newcommand{\eqa}{\begin{eqnarray}}
\newcommand{\eeqa}{\end{eqnarray}}
\newcommand{\beq}{\begin{equation}}
\newcommand{\eeq}{\end{equation}}

\newtheorem{dfn}{Definition}[section]
\newtheorem{thm}[dfn]{Theorem}

\newtheorem{lem}[dfn]{Lemma}

\newtheorem{prp}[dfn]{Proposition}

\def \= {\; =\; }
\def \:= {\, : =\, }

\newcommand{\+} {\; + \; }
\newcommand{\nn}{\nonumber}
\newcommand{\f}{\frac}
\newcommand{\e}{\epsilon}
\newcommand{\al}{\alpha}

\newcommand{\p}{\partial}
\newcommand{\F}{{\mathcal{F}}}

\begin{document}
\title{Hamiltonian perturbations at the second order approximation}
\author{Di Yang\,\footnote{diyang@ustc.edu.cn}\\
\small{School of Mathematical Sciences, USTC, Hefei 230026, P.R.~China}}
\date{}
\maketitle
\begin{abstract}
Integrability condition of Hamiltonian perturbations of 
integrable Hamiltonian PDEs of hydrodynamic type up to the 
second order approximation is considered.
\end{abstract}




\section{Introduction}
Let $M$ be an $n$-dimensional complex manifold. Consider the following system of Hamiltonian PDEs of hydrodynamic type:
\begin{equation}\label{dispersionless}
\p_t  \bigl(v^\alpha\bigr) \= \eta^{\alpha\beta} \p_x \biggl( \frac{\delta H_0}{\delta v^{\beta}(x)} \biggr) \, , 
\qquad v=(v^1,...,v^n)\in M,\, x\in S^1,\,t\in\mathbb{R} \, , 
\end{equation}
where $(\eta^{\alpha\beta})$ is a given symmetric invertible constant matrix,
$H_0 := \int_{S^1} h_0(v)\,dx$ is a given local functional (called the Hamiltonian), and 
$\delta/\delta v^\beta(x)$ denote the variational derivative. 
Here and below, free Greek indices take the integer values $1,\dots,n$, and the Einstein summation convention is assumed for repeated Greek indices with one-up and one-down; the matrix~$(\eta^{\al\beta})$ and its inverse~$(\eta_{\al\beta})$ 
are used to raise and lower Greek indices, e.g. 
$v_\alpha:=\eta_{\alpha\beta}v^{\beta}$.
The Hamiltonian density~$h_0(v)$ is assumed to be a holomorphic function of~$v$. More explicitly, 
equations~\eqref{dispersionless} have the form:
\[\p_t \bigl(v^\alpha\bigr) \= A^\alpha_\gamma(v)\,v^\gamma_x\,, \quad 
{\rm where} ~A^\alpha_\gamma(v) \:= \eta^{\alpha\beta} \frac{\p^2 h_0}{\p v^{\beta}\p v^{\gamma}}\,.\]

\noindent {\bf Basic assumption:} $(A^\alpha_\gamma(v))$ 
has pairwise distinct eigenvalues $\lambda_1(v),\dots,\lambda_n(v)$ on an open dense subset~$U$ of~$M$. 

Let us perform a change of variables $(v^1,\dots,v^n)\rightarrow (R_1,\dots,R_n)$ with non-degenerate Jacobian 
on~$U$. We call $R_1,\dots,R_n$ a complete set of Riemann invariants, if evolutions along $R_1,\dots,R_n$ are all diagonal, namely, 
\begin{equation} \label{rie-def}
\p_t (R_i) \= V_i(R) \, \p_x (R_i)\,, \qquad i=1,\dots,n\,,
\end{equation}
where $V_i$'s are some functions of $R=(R_1,...,R_n)$. 
Below, free Latin indices take the integer values $1,\dots,n$ unless otherwise indicated. Clearly, equations~\eqref{rie-def} imply that 
the gradients of Riemann invariants are eigenvectors of~$(A^{\alpha}_\beta)$, namely,
\begin{equation}
A^\alpha_\beta\,R_{i,\alpha} \= \lambda_i \,R_{i,\beta}\,,\qquad V_i \= \lambda_i
\end{equation}
with $R_{i,\alpha}:=\p_\alpha (R_i)$. Similar notations like $R_{i,j}:=\p_j(R_i)$, $R_{i,jk}:=\p_j \p_k (R_i)$, $\cdots$ will also be used. 
Here and below, $\p_\alpha:=\p_{v^\alpha}$, $\p_i:=\p_{R_i}$.

It was proven by Tsarev~\cite{Tsarev} that the integrability of equations~\eqref{dispersionless} is equivalent to the existence of complete Riemann invariants. 
Here, ``integrability" means existence of sufficiently many conservation laws/infinitesimal symmetries (See Definition~\ref{integrable}). 
It was shown by B.~Dubrovin~\cite{Du3, Du4} that existence of a complete set of Riemann invariants 
is equivalent to {\it vanishing} of the following Haantjes tensor:
\beq\label{Haa}
H_{\alpha\beta\gamma} \:= \bigl(A_{\alpha \rho\sigma} A_{\beta\phi} A_{\gamma \psi}+A_{\beta \rho\sigma} A_{\gamma\phi} A_{\alpha \psi}+A_{\gamma \rho\sigma} A_{\alpha\phi} A_{\beta \psi}\bigr) 
\, A^\rho_\nu \, \delta^{\sigma\nu\psi\phi} \, ,
\eeq
where $A_{\alpha\beta\gamma}:=\p_\alpha \p_\beta \p_\gamma (h_0)$ and
$\delta^{\alpha\beta\gamma\phi} := \eta^{\alpha\gamma} \eta^{\beta\phi}  \,-\,  \eta^{\alpha\phi} \eta^{\beta\gamma}$.
Note that $H_{\alpha\beta\gamma}$ automatically vanishes if the signature $\varepsilon(\alpha,\beta,\gamma)=0$; 
so for $n=1$ or for $n=2$, the system~\eqref{dispersionless} is always integrable. 

We proceed to the study of Hamiltonian perturbations \cite{DMS,Du1,Du2,Du3,Du4,DZ-norm,Getzler} of~\eqref{dispersionless} 
\begin{equation}\label{with_dispersion}
\p_t (v^\alpha) \= \eta^{\alpha\beta} \p_x \biggl(\frac{\delta H}{\delta v^{\beta}(x)}\biggr)\,,\qquad x\in S^1\,, ~ t\in\mathbb{R}\,, ~ v=(v^1,\dots,v^n)\in M \,.
\end{equation}
Here,  $ H := \int_{S^1} h\,dx = \sum_{j=0}^\infty \e^j H_j$ is the Hamiltonian with $H_j := \int_{S^1} h_j(v,v_1,v_2,...,v_j)\,dx$,    
and $h_j$ are differential polynomials of~$v$ satisfying the following homogeneity condition:
\beq
\sum_{\ell=1}^j \ell\,v^\alpha_\ell \frac{\p h_j}{\p v^\alpha_\ell} \= j\, h_j \,, \quad j\geq 0\,.
\eeq
We recall that the variational derivative reads
$$
\frac{\delta H}{\delta v^{\beta}(x)} \=  \sum_{\ell=0}^\infty (-\p_x)^\ell \biggl(\frac{\p h}{\p v^{\beta}_\ell}\biggr) \,.
$$
In the above formulae, $v_\ell^\alpha:=\p_x^\ell (v^\alpha)$, $\ell\geq 0$, 
and we recall that a differential polynomial of~$v$ is a polynomial of $v_1,v_2,\dots$ whose coefficients are holomorphic functions of~$v$.
The ring of differential polynomials of $v$ is denoted by $\mathcal{A}_v$. 
Note that the Hamiltonian operator $\eta^{\alpha\beta}\p_x$ defines a Poisson bracket $\{\,,\,\}$ on the space of local functionals $\mathcal{F}:=\bigl\{\int_{S^1} f \,dx \,| \, f\in \mathcal{A}_v[[\e]]\bigr\}$, 
$\{\,,\,\}: \mathcal{F}\times \mathcal{F} \rightarrow \mathcal{F}$, by
\begin{equation}\label{Poisson}
\{F,G\} \:= \int_{S^1} \frac{\delta F}{\delta v^{\alpha}(x)}\eta^{\alpha\beta}\p_x \biggl(\frac{\delta G}{\delta v^{\beta}(x)}\biggr)\,dx \,, \qquad \forall~ F,G\in \F \,.
\end{equation}
It is helpful to view $v^\alpha(x)$ as a ``local functional" 
$v^\alpha(x)=\int_{S^1} v^\alpha(y)\, \delta(y-x)\, dy$, called the coordinate functional. Then one can write 
equation~\eqref{with_dispersion} in the form
$$
\p_t (v^\alpha) \= \bigl\{ v^\alpha(x) \,, \, H\bigr\} \,.
$$
Clearly, a system of Hamiltonian PDEs of hydrodynamic type~\eqref{dispersionless} can be obtained from~\eqref{with_dispersion} simply 
by taking the dispersionless limit: $\epsilon\rightarrow 0$. 

The perturbed system~\eqref{with_dispersion} is called \emph{integrable} if its dispersionless limit is integrable and each 
conservation law of~\eqref{dispersionless} can be extended to a conservation law of~\eqref{with_dispersion}. 
In this paper, we start with a system of {\it integrable Hamiltonian PDEs of hydrodynamic type}, and 
study the conditions such that the perturbation~\eqref{with_dispersion} is integrable up to the second order approximation. 

\begin{thm}\label{thm10}
Assume that \eqref{dispersionless} is integrable.  Denote by~$U$ the open dense subset under consideration, by $\lambda_1,\dots,\lambda_n$ the 
distinct eigenvalues of~$(A^{\alpha}_\beta)$, and by $R=\bigl(R_1,\dots,R_n\bigr)$ the associated complete Riemann invariants. 
A Hamiltonian perturbation of~\eqref{dispersionless} of the form
$H = H_0+\e\,H_1+\mathcal{O}(\e^2)$
with $H_0 = \int_{S^1} h(v)\,dx$, $H_1 = \int_{S^1} \sum_{i=1}^n p_i(R)\,{R_i}_x\, dx$
is integrable at the first order approximation iff
\begin{equation}\label{first-order-cond}
\omega_{ij,k} \,-\, \omega_{ik,j} \= a_{ij}\,\omega_{ik} \+ a_{ji}\,\omega_{jk} \,- \, a_{ik}\,\omega_{ij} \,- \, a_{ki}\,\omega_{kj} \,,\qquad\forall\,\varepsilon(i,j,k)=\pm1 \,.
\end{equation}
Here, $a_{ij}$ and $\omega_{ij}$ are defined by
\beq
a_{ij} \:= \frac{\lambda_{i,j}}{\lambda_i-\lambda_j}\,, \quad \omega_{ij} \:= \frac{p_{i,j}-p_{j,i}}{\lambda_i-\lambda_j}\,,\qquad\forall~i\neq j \, .
\eeq
Assume that  $\lambda_{i,i}(v) \neq 0$ for $v\in U$. If a Hamiltonian perturbation of~\eqref{dispersionless}  
of the form
\begin{equation}
H \= H_0 \+ \e^2\,H_2 \+ \mathcal{O}(\e^3) \label{sec}
\end{equation}
with
$H_0 = \int_{S^1} h_0(u)\,dx$, $H_2 = \int_{S^1} \sum_{i,j=1}^n d_{ij}(R)\, {R_i}_x {R_j}_x\, dx$ $(d_{ij}=d_{ji})$
is integrable, then there exist functions~$C_i(R_i)$ such that 
\begin{align}
& d_{ii} \= - \, C_i(R_i) \, \lambda_{i,i} \, , \label{d1} \\
& \biggl(\frac{d_{ij}}{\lambda_i-\lambda_j}\biggr)_{,\,k} \+ \biggl(\frac{d_{jk}}{\lambda_j-\lambda_k}\biggr)_{,\,i} \+ \biggl(\frac{d_{ki}}{\lambda_k-\lambda_i}\biggr)_{,\,j} \= 0\,,\qquad \forall~\varepsilon(i,j,k)=\pm 1 \,. \label{d2}
\end{align}
\end{thm}
\noindent An equivalent description of~\eqref{d1}--\eqref{d2} is that the density~$h_2$ can be written in the form
\begin{equation}\label{sec-h2}
h_2 \= - \, \sum_{i=1}^n C_i(R_i)\, \lambda_{i,i} \, {R_i}_x^2 \+ \frac{1}{2}\sum_{i\neq j}\, (\lambda_i-\lambda_j)\, s_{ij}\, {R_i}_x \, {R_j}_x \, ,
\end{equation}
where $s_{ij}:=\phi_{i,j}-\phi_{j,i}$ for some functions~$\phi_i(R)$.

One important tool of studying Hamiltonian perturbations is to use Miura-type and quasi-Miura transformations~\cite{DZ-norm}. 
Recall that a Miura-type transformation near identity is given by an invertible map of the form
\begin{equation}\label{Miura}
v\mapsto w\,, \qquad w^\alpha \:= \sum_{j=0}^\infty \, \e^{j} \, W^\alpha_j(v,v_1,...,v_\ell) \,, ~ W^\alpha_0 = v^\alpha \,, 
\end{equation}
where $W^\alpha_j$, $j\geq 0$ are differential polynomials of~$v$ homogeneous of degree~$j$ with respect to the 
degree assignments $\deg v^\alpha_\ell = \ell$,~$\ell\geq 1$.
A Miura-type transformation is called \emph{canonical} if there exists a local functional~$K$, such that 
\begin{equation}\label{cm}
w^\alpha \= v^\alpha \+ \epsilon\, \bigl\{v^{\alpha}(x),K \bigr\} \+ \frac{\epsilon^2}{2!}\,\bigl\{\!\bigl\{v^{\alpha}(x),K \bigr\},K\bigr\} \+ \cdots
\end{equation}
where $K=\sum_{j=0}^\infty \e^j K_j$. Two Hamiltonian perturbations of the same form~\eqref{with_dispersion} are called \emph{equivalent} if they are related via a canonical Miura-type transformation.
A Hamiltonian perturbation~\eqref{with_dispersion} is called \emph{trivial} if it is equivalent to~\eqref{dispersionless}. 

A map of the form~\eqref{Miura} is called a {\it quasi-Miura} transformation, if $W_\ell^\alpha$,~$\ell\geq 1$ are 
allowed to have rational and logarithmic dependence in~$v_x$. 
The Hamiltonian perturbation~\eqref{with_dispersion} 
is called {\it quasi-trivial} or possessing \emph{quasi-triviality}, if it is related via a canonical quasi-Miura transformation to~\eqref{dispersionless}.
The precise definition used in this paper for quasi-Miura transformation is given in Section~3.
Many interesting nonlinear PDE systems possess quasi-triviality; for example, it was shown in~\cite{DLZ} that 
if~\eqref{with_dispersion} is {\it bihamiltonian} then it is quasi-trivial.

The existence of quasi-triviality of~\eqref{with_dispersion} at the second order approximation 
and its relationship with integrability together with an application are described in the following theorem.

\begin{thm}\label{thm1}
Assume \eqref{dispersionless} is integrable and $\lambda_{i,i}$ is nowhere  zero in~$U$. {\rm Part I}. Up to the second order approximation, the Hamiltonian perturbation~\eqref{with_dispersion}
is quasi-trivial iff it is integrable. {\rm Part II}.  
If $H_2$ has the form~\eqref{sec-h2}, then the Hamiltonian perturbation~\eqref{sec} is $\mathcal{O}(\e^2)$-integrable.
\end{thm}
For the cases $n=1,2$, Theorem~\ref{thm1} agrees with the results of~\cite{LZ} and~\cite{Du2}.

\medskip

The article is organized as follows. In Section~\ref{sect2}, we review some terminologies about Hamiltonian PDEs. 
In Section~\ref{sect3}, we study integrability of~\eqref{with_dispersion} up to the second order approximation. 
An example of non-integrable perturbation is given in Section~\ref{sect4}.

\section{Preliminaries}\label{sect2}
In this section, we will recall several terminologies in the theory of Hamiltonian perturbations; more 
terminologies can be found in e.g.~\cite{DZ-norm, Tsarev, DLZ, DuII, Du6, Dulecture,Du3, Serre}.
\begin{dfn}
A local functional $F_0 = \int_{S^1} f_0(v)\,dx$
is called a \emph{conserved quantity} of~\eqref{dispersionless} if
\begin{equation}\label{con1}
\frac{dF_0}{dt} \= 0 \,.
\end{equation}
Here the density $f_0(v)$ is a holomorphic function of~$v$.
\end{dfn}
\noindent We also often call a conserved quantity a conservation law. 
Note that for simplicity we will exclude the degenerate ones with $f_0(v)\equiv{\rm const}$ 
from conservation laws.

Since~\eqref{dispersionless} is a Hamiltonian system, equation~\eqref{con1} can be written equivalently as
\beq
\bigl\{H_0,F_0\bigr\} \= 0 \,,
\eeq
where $\{\,,\,\}$ denotes the Poisson bracket defined in~\eqref{Poisson}. According to Noether's theorem, \eqref{con1} is 
also equivalent to the statement that the following Hamiltonian flow generated by~$F_0$
$$
v^\alpha_s\:= \{v^\alpha(x), F_0\}
$$
commutes with~\eqref{dispersionless}. 
Let $(M_{\alpha\beta})$ denote the Hessian of~$f$, i.e. $M_{\alpha\beta}:=\p_\alpha \p_\beta (f)$.
Equation~\eqref{con1} then reads
\begin{equation} \label{con3}
A^\alpha_\gamma \, M^\gamma_\beta \= M^\alpha_\gamma \, A^\gamma_\beta.
\end{equation}
\begin{dfn}\label{integrable}
The PDE system~\eqref{dispersionless} is called integrable if it possesses an infinite family of conserved quantities  
parametrized by~$n$ arbitrary functions of one variable.
\end{dfn}
A necessary and sufficient condition for integrability of~\eqref{dispersionless} is 
the vanishing of the Haantjes tensor $H_{\alpha\beta\gamma}$~\eqref{Haa} as recalled already in the introduction.
We will assume that \eqref{dispersionless} is integrable and study its perturbations. Recall that vanishing of the Haantjes tensor 
ensures the existence of a complete set of Riemann invariants $\{R_1,...,R_n\}$.  We have
\begin{align}
\label{eigen-eigen}
& A^\alpha_\beta\,R_{i,\alpha}  \=  \lambda_i \,R_{i,\beta} \, ,\\
& M^\alpha_\beta\,R_{i,\alpha} \= \mu_i \,R_{i,\beta} \, .\label{eigen-eigen2}
\end{align}
Here, $\mu_i$ are eigenvalues of $(M^\alpha_\beta)$. 
For a generic conserved quantity $F_0$, the eigenvalues $\mu_1,...,\mu_n$ on the~$U$ are also pairwise distinct. 
In terms of $\lambda_i,\mu_i$ the flow commutativity is equivalent to
\begin{equation}\label{con4}
a_{ij} \= b_{ij}\,,\qquad \forall\,i\neq j \, ,
\end{equation}
where
\begin{equation}\label{ab}
a_{ij} \:= \frac{\lambda_{i,j}}{\lambda_i-\lambda_j} \, , \qquad b_{ij} \:= \frac{\mu_{i,j}}{\mu_i-\mu_j} \, .
\end{equation}
The compatibility condition 
$$
\mu_{i,jk} \= \mu_{i,kj}\,,\qquad\forall\,\varepsilon(i,j,k)=\pm1
$$
for equations \eqref{con4} reads as follows
\begin{equation}\label{comp}
(\mu_i-\mu_k)(a_{ij,k}-a_{ik,j}) \,-\, (\mu_j-\mu_k)(a_{ij,k}+a_{ij}a_{jk}+a_{ik}a_{kj}-a_{ij}a_{ik}) \= 0 \,.
\end{equation}
Definition~\ref{integrable} requires that equation~\eqref{comp} is true for infinitely many $F_0$ 
parametrized by~$n$ arbitrary functions of one variable. So the coefficients of $\mu_i-\mu_k$ and of $\mu_j-\mu_k$ must vanish:
\begin{align}
& a_{ij,k} \,-\, a_{ik,j} \= 0 \, ,\qquad\forall\,\varepsilon(i,j,k)=\pm1 \, ,\label{Tsarev}\\
& a_{ij,k} \+ a_{ij}a_{jk} \+ a_{ik}a_{kj} \,-\, a_{ij}a_{ik} \= 0 \, ,\qquad\forall\,\varepsilon(i,j,k)=\pm1 \, .\label{coTs}
\end{align}
Note that~\eqref{coTs} is implied by equations~\eqref{Tsarev} and~\eqref{ab}.

\begin{dfn}\label{condef}
A local functional $F:= \sum_{j=0}^\infty \e^j F_j$
is called a conserved quantity of~\eqref{with_dispersion}, if
\begin{equation}
\label{int}\frac{d F}{dt} \= 0 \,.
\end{equation}
Here, $F_j=\int_{S^1} f_j(v,v_1,...,v_j)\,dx$, $j\geq 0$ with $f_j$ being differential polynomials of~$v$ homogeneous of degree~$j$.
\end{dfn}
\noindent Conserved quantities (or say conservation laws) considered in this paper are always of the form as in Definition~\ref{condef}.

Equation~\eqref{int} can be equivalently written as 
$$
\{H,F\} \= 0 \, ,
$$
which is recast into an infinite sequence of equations 
\begin{align}
& \{H_0,F_0\} \= 0 \,, \nn \\
& \{H_0,F_1\} \+ \{H_1,F_0\} \= 0 \,, \nn \\
& \{H_0,F_2\} \+ \{H_1,F_1\} \+ \{H_2,F_0\} \= 0 \,, \nn\\
& \, \mbox{etc.}\,.\nn
\end{align}
\begin{dfn}
A Hamiltonian perturbation~\eqref{with_dispersion} is called integrable if its dispersionless limit~\eqref{dispersionless} is integrable 
and generic conservation laws of~\eqref{dispersionless} can be extended to those of~\eqref{with_dispersion}. For $N\geq1$, 
\eqref{with_dispersion} is called $\mathcal{O}(\e^N)$-integrable if its dispersionless limit \eqref{dispersionless} is integrable 
and every generic conservation law $F_0$ of~\eqref{dispersionless} can be extended to a local functional~$F$, s.t.
\beq
\{H,F\} \= \mathcal{O}(\e^{N+1}).
\eeq
\end{dfn}

\section{Proofs of Theorems~\ref{thm10} and~\ref{thm1}} \label{sect3}

In this section, we study integrability of the Hamiltonian system~\eqref{with_dispersion} up to the second order approximation, 
and prove Theorems~\ref{thm10} and~\ref{thm1}.

Assume that \eqref{dispersionless} is integrable. 

We start with the first order approximation.
Let us first look at the integrability condition of the $\mathcal{O}(\e^1)$-approximation.
Denote
\begin{equation}
H \= H_0 \+ \e\,H_1 \+ \mathcal{O}(\e^2) \label{pert-1}
\end{equation}
with $H_1 = \int_{S^1} \tilde p_\alpha(u)\, u^\alpha_x\,dx = \sum_{i=1}^n \int_{S^1} p_i(R)\, {R_i}_x\,dx$.
Here, the functions $p_\alpha$ and $p_i$ are assumed to satisfy
$\tilde p_\alpha =\sum_{i=1}^n p_i R_{i,\alpha}$.

\noindent{\it Proof} of Theorem~\ref{thm10}.  \quad 
Denote by $\tilde \theta_{\alpha\beta}$ the exterior differential of the $1$-form $\tilde p_\alpha du^\alpha$
\beq
\tilde \theta_{\alpha\beta} \= \tilde p_{\alpha,\beta} \,-\, \tilde p_{\beta,\alpha} \, .
\eeq
In the coordinate chart of the Riemann invariants $R_1,...,R_n$, we have
$$
\theta_{ij} \= \p_i u^{\alpha}\, \tilde \theta_{\alpha\beta}\,\p_j u^{\beta} \= p_{i,j} \,-\, p_{j,i} \, .
$$
The $\mathcal{O}(\e^1)$-integrability says any local functional $F_0=\int_{S^1} f(u)\,dx$ satisfying
$$
\{H_0,F_0\} \= 0
$$
can be extended to a local functional
$$
F \= F_0 \+ \e\,F_1 \+ \mathcal{O}(\e^2)
$$
such that
\begin{equation}
\label{first-int-short}
\{H,F\} \= \mathcal{O}(\e^2).
\end{equation}
Here, the local function $F_1$ is of the form 
\beq
F_1 \= \int_{S^1} \tilde q_\alpha(u)\,u^\alpha_x\,dx \= \sum_{i=1}^n \int_{S^1} q_i(R)\,{R_i}_x\,dx \, .
\eeq
Equation~\eqref{first-int-short} reads as follows
$$ \{H_0,F_1\} \+ \{H_1,F_0\} \= 0 \, , $$
which is equivalent to
\begin{equation}\label{first-int-long1}
\tilde \theta_{\alpha\gamma}M^\gamma_\beta \+ \tilde \theta_{\beta\gamma} M^\gamma_\alpha \= 
\tilde \Theta_{\alpha\gamma}A^\gamma_\beta \+ \tilde \Theta_{\beta\gamma}A^\gamma_\alpha
\end{equation}
or, in the coordinate system of the Riemann invariants, to
\begin{equation}
\label{first-int-long2}\frac{\theta_{ij}}{\lambda_i-\lambda_j} \= \frac{\Theta_{ij}}{\mu_i-\mu_j} \,, \qquad \forall~i\neq j.
\end{equation}
Here, $\tilde \Theta_{\alpha\beta}:= \tilde q_{\alpha,\beta}-\tilde q_{\beta,\alpha}$, $\Theta_{ij}:=q_{i,j}-q_{j,i}$. The compatibility condition of~\eqref{first-int-long2} is given by
\beq\label{comp-th}
\Theta_{ij,k} \+ \Theta_{jk,i} \+ \Theta_{ki,j} \= 0 \, ,\quad \forall~\varepsilon(i,j,k)=\pm 1.
\eeq
Introduce the notations 
\beq
\omega_{ij} \= \frac{\theta_{ij}}{\lambda_i-\lambda_j} \,,\qquad i\neq j \, .\eeq
Then equations \eqref{comp-th} imply
$$
\p_k\,[\,\omega_{ij}\,(\mu_i-\mu_j)\,]+\p_i\,[\,\omega_{jk}\,(\mu_j-\mu_k)\,]+\p_j\,[\,\omega_{ki}\,(\mu_k-\mu_i)\,] \= 0 \,,\qquad \forall~\varepsilon(i,j,k)=\pm 1 \, ,
$$
i.e.
\beq \label{omega-cyc}
\omega_{ij,k}\,(\mu_i-\mu_j) \+ \omega_{ij}\,(\mu_{i,k}-\mu_{j,k}) \+ \mbox{cyclic} \= 0 \,,\qquad \forall~\varepsilon(i,j,k)=\pm 1.
\eeq
Substituting equations \eqref{con4},\,\eqref{ab} in equations \eqref{omega-cyc} we obtain
\beq
\omega_{ij,k}\,(\mu_i-\mu_j) \+ \omega_{ij}\,(a_{ik}(\mu_i-\mu_k)-a_{jk}(\mu_j-\mu_k)) \+ \mbox{cyclic} \= 0 \, ,
\eeq
from which we obtain that for any pairwise distinct $i,j,k$,
\begin{align}
& (\mu_i-\mu_k) \bigl(\omega_{ij,k}+\omega_{ij}\,a_{ik}-\omega_{jk}\,a_{ji}+\omega_{jk}\,a_{ki}-\omega_{ki,j}-\omega_{ki}\,a_{ij}\bigr)\nn\\
& \+ (\mu_j-\mu_k) \bigl(-\omega_{ij,k}-\omega_{ij}\,a_{jk}+\omega_{jk,i}+\omega_{jk}\,a_{ji}-\omega_{ki}\,a_{kj}+\omega_{ki}\,a_{ij}\bigr) \= 0 \, .
\end{align}
As a result we conclude that
\begin{align}
& \omega_{ij,k} \+ \omega_{ij}\,a_{ik} \,-\, \omega_{jk}\,a_{ji} \+ \omega_{jk}\,a_{ki}-\omega_{ki,j} \,-\, \omega_{ki}\,a_{ij} \= 0 \, ,\quad\forall\,\varepsilon(i,j,k) =\pm1 \, , \label{repeat1}\\
& - \omega_{ij,k} \,-\, \omega_{ij}\,a_{jk} \+ \omega_{jk,i} \+ \omega_{jk}\,a_{ji}-\omega_{ki}\,a_{kj} \+ \omega_{ki}\,a_{ij} \= 0 \, ,\quad\forall\,\varepsilon(i,j,k)=\pm1 \,. \label{repeat2}
\end{align}
The part on the first order approximation of the theorem is proved. 

In below we will continue to
prove the second part.

Let us consider the condition of quasi-triviality at the first order of approximation. 
We call the Hamiltonian perturbation~\eqref{pert-1} is {\it quasi-trivial}, if there exists a local functional 
$$
K_0 \= \int_{S^1} k_0(v)\,dx
$$
such that
\begin{equation}\label{exist-trivial}
\{H_0,K_0\} \= H_1 \, .
\end{equation}
Clearly, our definition of quasi-triviality at the first order approximation is the same as triviality. 
Equation~\eqref{exist-trivial} is equivalent to the existence of a function $\psi$ satisfying 
\begin{equation}\label{eqv-e-t-1}
\tilde p_\alpha \= \frac{\p k_0}{\p u^\gamma}A^\gamma_\alpha \+ \frac{\p \psi}{\p u^\alpha} \, .
\end{equation}
Eliminating $\psi$ in the above equation we find the following equivalent equation to~\eqref{exist-trivial}:
\begin{equation}\label{eqv-e-t-2}
\tilde \theta_{\alpha\beta} \= \frac{\p^2 k_0}{\p u^\beta\p u^\gamma} A^\gamma_\alpha \,-\, \frac{\p^2 k_0}{\p u^\alpha\p u^\gamma} A^\gamma_\beta \, .
\end{equation}
In the coordinate chart of Riemann invariants, equations \eqref{eqv-e-t-1} and \eqref{eqv-e-t-2} become
\begin{align}
\label{eqv-e-t-Rie-1}  & p_i \= \lambda_i\, k_{0,i}+\psi_{,i} \, ,\\
\label{eqv-e-t-Rie-2}  & \frac{\theta_{ij}}{\lambda_i-\lambda_j} \= k_{0,ij} \+ a_{ij}\,k_{0,i} \+ a_{ji}\,k_{0,j} \,,\qquad i\neq j \, .
\end{align}

\begin{prp}\label{Propo}
Assume~\eqref{dispersionless} is integrable. At the first order approximation, the Hamiltonian perturbation \eqref{pert-1}
is (quasi-)trivial iff it is integrable.
\end{prp}
\begin{proof}
The compatibility condition of equation \eqref{eqv-e-t-Rie-2} is given by
$$
\p_k k_{0,ij} \= \p_j k_{0,ik} \,, \qquad\forall~\varepsilon(i,j,k)=\pm1,
$$
which yields
\begin{equation}\label{comp2}
\p_k\biggl(\frac{\theta_{ij}}{\lambda_i-\lambda_j}-a_{ij}\,k_{0,i}-a_{ji}\,k_{0,j}\biggr) \=
\p_j\biggl(\frac{\theta_{ik}}{\lambda_i-\lambda_k}-a_{ik}\,k_{0,i}-a_{ki}\,k_{0,k}\biggr) \, .
\end{equation}
Substituting equation~\eqref{eqv-e-t-Rie-2} into~\eqref{comp2} we find
\begin{align}
&\omega_{ij,k} \,-\, a_{ij}\,\omega_{ik} \,-\, a_{ji}\,\omega_{jk} \,-\, k_{0,i}\,a_{ij,k} \+ k_{0,j}(a_{ji}a_{jk}-a_{ji,k}) \+ k_{0,k}(a_{ki}a_{ij}+a_{kj}a_{ji})\nn\\
& \= \omega_{ik,j} \,-\, a_{ik}\,\omega_{ij} \,-\, a_{ki}\,\omega_{kj} \,-\, k_{0,i}\,a_{ik,j} \+ k_{0,k}(a_{ki}a_{kj}-a_{ki,j}) \+ k_{0,j}(a_{ji}a_{ik}+a_{jk}a_{ki}) \, .
\label{320}
\end{align}
Finally substituting equations~\eqref{Tsarev} and~\eqref{coTs} into~\eqref{320}, we have
\beq
\omega_{ij,k} \,- \, a_{ij}\,\omega_{ik} \,-\, a_{ji}\,\omega_{jk} \= \omega_{ik,j} \,-\, a_{ik}\,\omega_{ij} \,-\, a_{ki}\,\omega_{kj} \, ,\quad\forall~\varepsilon(i,j,k)=\pm1,
\eeq
which finishes the proof.
\end{proof}

We now proceed to the second order approximation. Let 
\begin{equation}\label{pert-2}
H \= H_0 \+ \e \, H_1 \+ \e^2 \, H_2 \+ \mathcal{O}(\e^3)
\end{equation}
be a Hamiltonian perturbation of~\eqref{dispersionless} with 
\beq
H_2 \= \int_{S^1} \tilde d_{\alpha\beta}(v)\,v^\alpha_x \,v^\beta_x\,dx \= \int_{S^1} \sum_{i,j=1}^n d_{ij} {R_i}_x {R_j}_x \, dx
\eeq
and 
\beq
\tilde d_{\alpha\beta} \= \tilde d_{\beta\alpha} \, ,\qquad d_{ij} \= d_{ji} \:= \tilde d_{\alpha\beta} \,v^{\alpha}_{,\,i} \,  v^\beta_{,\,j} \, .
\eeq
Assume as always that \eqref{dispersionless} is integrable, and assume that \eqref{pert-2}
is $\mathcal{O}(\e^1)$-integrable. According to Proposition~\ref{Propo}, 
there exists a canonical Miura-type transformation
reducing $H_1$ to the zero functional. 
So the assumption that $H_1 = 0$ used in~\eqref{sec} in the statement of Theorem~\ref{thm10} does not lose generality.

Let us study the necessary condition of $\mathcal{O}(\e^2)$-integrability for~\eqref{pert-2}. 

\noindent{\it Continuation of Proof} of Theorem~\ref{thm10}.  \quad 
Recall that the $\mathcal{O}(\e^2)$-integrability means that, for a generic conservation law~$F_0$ of~\eqref{dispersionless}, 
there exists a local functional of the form
\beq
F_2 \= \int_{S^1} \tilde D_{\alpha\beta}(u)\,u^\alpha_x\,u^\beta_x\,dx \= \sum_{i,j=1}^n \int_{S^1} D_{ij}(R) {R_i}_x {R_j}_x \,dx
\eeq
such that 
\beq  \label{int-hf-2} \{H_0,F_2\} \+ \{H_2,F_0\} \= 0 \, . \eeq

Note that equation~\eqref{int-hf-2} implies
\begin{align}
& M^\rho_\sigma\, \tilde d_{\rho\beta}\,-\,M^\rho_\beta\, \tilde d_{\rho\sigma} 
\= A^\rho_\sigma\, \tilde D_{\rho\beta}\,-\,A^\rho_\beta\, \tilde D_{\rho\sigma}\,,\label{int-(1)}\\
& M^\rho_\gamma\, \tilde d_{\rho\sigma,\beta} \+ M^\rho_\sigma\, \tilde d_{\rho\beta,\gamma} \+ M^\rho_\beta\, \tilde d_{\rho\gamma,\sigma} 
\,-\, M^\rho_{\sigma\gamma}\, \tilde d_{\rho\beta}\,-\,M^\rho_{\sigma\beta}\, \tilde d_{\rho_\gamma}\,-\,M^\rho_{\beta\gamma}\, \tilde d_{\rho\sigma}\nn\\
&\qquad\qquad\qquad\qquad\qquad -\,M^{\rho}_\sigma\, \tilde d_{\beta\gamma,\rho} 
\,-\,M^{\rho}_{\beta}\, \tilde d_{\sigma\gamma,\rho}\,-\,M^{\rho}_\gamma \, \tilde d_{\sigma\beta,\rho}\nn\\
& \= (M\leftrightarrow A, d\leftrightarrow D) \, . \label{int-(2)}
\end{align}
In the coordinate system of the complete Riemann invariants, \eqref{int-(1)} and~\eqref{int-(2)} become
\begin{align}
& \frac{D_{ij}}{\mu_i-\mu_j} \= \frac{d_{ij}} {\lambda_i-\lambda_j} \, ,\quad \forall~i\neq j, \label{ij}\\
& \lambda_{i,l}D_{ij} \+ \lambda_{j,i}D_{jl} \+ \lambda_{i,j} D_{il} \+ (\lambda_i-\lambda_l)D_{lj,i} \+ (\lambda_j-\lambda_l)D_{li,j} \+ (\lambda_l-\lambda_j)D_{ij,l}\nn\\
& \= \mu_{i,l}\,d_{ij} \+ \mu_{j,i}\,d_{jl} \+ \mu_{i,j}\, d_{il} \+ (\mu_i-\mu_l)\,d_{lj,i} \+ (\mu_j-\mu_l)\,d_{li,j} \+ (\mu_l-\mu_j)\,d_{ij,l}\,,\quad \forall~i,j,l \,. \label{ijk} 
\end{align}
Note that in the derivation of~\eqref{ijk} we have used~\eqref{ij}.

Taking $j=l=i$ in~\eqref{ijk} we obtain
\beq\label{(1)}
\lambda_{i,i}\, D_{ii} \= \mu_{i,i}\, d_{ii} \,.
\eeq
By assumption, in the subset~$U$ of~$M$, $\lambda_i$ satisfy
$\lambda_{i,i}\neq 0$.
Thus there exist functions~$C_i(R)$ such that
\beq\label{dii}
D_{ii} \= - \, C_i (R)\, \mu_{i,i} \,, \qquad d_{ii} \= - \, C_i(R) \, \lambda_{i,i} \, .
\eeq
Taking $l=j$ and $i\neq j$ in~\eqref{ijk} we find
\beq\label{diij}
\lambda_{j,i}D_{jj} \+ (\lambda_i-\lambda_j) D_{jj,i} \= \mu_{j,i}\,d_{jj} \+ (\mu_j-\mu_i)\,d_{jj,i} \,,\qquad \forall~j\neq i.
\eeq
Substituting~\eqref{dii} into~\eqref{diij} and using~\eqref{Tsarev} we obtain
\beq
C_{j,i} \, \bigl((\lambda_i-\lambda_j)\mu_{j,j}-(\mu_i-\mu_j)\lambda_{j,j} \bigr) \= 0 \,,\qquad \forall\, j\neq i \, ,
\eeq
which implies
$$
C_{j,i} \= 0 \, ,\qquad \forall\,j\neq i \, ,
$$
i.e.
$$
C_j(R) \= C_j(R_j) \, .
$$
Taking $l=i$ and $j\neq i$ in~\eqref{ijk} and using~\eqref{(1)},\eqref{diij} we find
\beq\label{(3)}
\lambda_{i,i} D_{ij} \+ (\lambda_i-\lambda_j) D_{ij,i} \= \mu_{i,i} \, d_{ij} \+ (\mu_i-\mu_j) \, d_{ij,i} \, .
\eeq
Taking $j=i$ and $l\neq i$ in~\eqref{ijk} and using~\eqref{diij}, we find
\beq
\lambda_{i,i} D_{li} \+ (\lambda_{i}-\lambda_l)D_{li,i} \= \mu_{i,i}\,d_{li} \+ (\mu_i-\mu_l)\,d_{li,i}\,,
\eeq
which coincides with~\eqref{(3)}. It is straightforward to check that~\eqref{ij} and~\eqref{Tsarev} imply~\eqref{(3)}. So~\eqref{(3)} does not give new constraints to~$d_{ij}$, $i\neq j$.

Now we use~\eqref{ijk} in the case $\varepsilon(i,j,l)=\pm 1$. First, by~\eqref{ij} it is convenient to write
\beq\label{anti-symm}
D_{ij} \= s_{ij}(\mu_i-\mu_j) \, , \quad d_{ij} \= s_{ij} (\lambda_i-\lambda_j) \,,\qquad i\neq j \,,
\eeq
where $s_{ij}$ are some anti-symmetric fields.
Substituting~\eqref{anti-symm} in~\eqref{ijk} and using \eqref{Tsarev} we obtain
\beq
(s_{lj,i}+s_{ji,l}+s_{il,j})\left((\lambda_i-\lambda_l)(\mu_j-\mu_l)-(\lambda_j-\lambda_l)(\mu_i-\mu_l)\right) \= 0 \,, \quad \forall~\varepsilon(i,j,l)=\pm 1 \, .
\eeq
Hence
\beq
s_{lj,i} \+ s_{ji,l} \+ s_{il,j} \= 0 \,,  \qquad \forall~\varepsilon(i,j,l)=\pm 1 \,.
\eeq
The theorem is proved. $\hfill\square$

We now consider the condition of quasi-triviality for the Hamiltonian perturbation~\eqref{pert-2} with $H_1=0$.
Such a perturbation is called {\it quasi-trivial} if there exists a local functional~$K$ of the form
$$
K \= \e\, K_1 \+ \mathcal{O}(\e^2), \quad  K_1=\int_{S^1} k_1(u;u_x)\,dx,
$$
such that
\begin{equation}\label{deform-2}
H_0 \+ \e\,\{H_0,K\} \= H \, .
\end{equation}
Here $k_1$ is also required to satisfy the following homogeneity condition:
\begin{equation}\label{basic22}
\sum_{r\geq 1} r \,u^{\alpha}_r\frac{\p }{\p u^{\alpha}_r}\biggl(\frac{\p k_1}{\p u^\beta}-\p_x\biggl(\frac{\p k_1}{\p u^\beta_x}\biggr)\biggr)=\,\frac{\p k_1}{\p u^\beta}-\p_x\biggl(\frac{\p k_1}{\p u^\beta_x}\biggr).
\end{equation}

Equation~\eqref{basic22} is equivalent to the following linear PDE system:
\begin{align}
& u^\alpha_x\,k_{1,u^\alpha_x u^\beta_x u^\gamma_x} \+ k_{1,u^\beta_x u^\gamma_x} \= 0\,, \label{homo-1}\\
& u^\alpha_x\,k_{1,u^\alpha_x u^\beta} \,-\, u^\alpha_x u^\gamma_x\,k_{1,u^\gamma u^\alpha_x u^\beta_x}\,-\,k_{1,u^\beta} \= 0\,.\label{homo-2}
\end{align}
From equation~\eqref{deform-2} we obtain
$$
\{H_0,K_1\} \= H_2,
$$
which is equivalent to
\begin{equation}\label{basic}
\frac{\delta }{\delta u^\rho(x)}\biggl(H_2+\int_{S^1} \frac{\delta K_1}{\delta u^\alpha(x)} A^\alpha_\gamma u^\gamma_x\,dx\biggr) \= 0 \, .
\end{equation}
Equations \eqref{basic} read explicitly as follows
\begin{equation}\label{basic2}
\sum_{j=0}^2 (-1)^j \p_x^j \frac{\p}{\p u^{\rho}_j}\biggl[ \tilde d_{\alpha\beta}\,u^\alpha_x\, u^\beta_x+A^\alpha_\gamma\, u^\gamma_x\,\biggl(\frac{\p k_1}{\p u^\alpha}-\p_x\Bigl(\frac{\p k_1}{\p u^\alpha_x}\Bigr)\biggr)\biggr] \= 0 \, .
\end{equation}

We will now prove Theorem~\ref{thm1} by solving equations \eqref{homo-1},\,\eqref{homo-2},\,\eqref{basic2}.

\noindent{\it Proof} of Theorem~\ref{thm1}. \quad Part I. We first prove that quasi-triviality implies integrability. 

Comparing the coefficients of $u^{\sigma}_{xxx}$ of both sides of equations \eqref{basic2} gives
\begin{equation}\label{u3}
A^\alpha_\rho \, k_{1,u^\alpha_x u^\sigma_x} \= A^\alpha_\sigma \,k_{1,u^\alpha_x u^\rho_x }.
\end{equation}
In terms of the Riemann invariants, equations \eqref{u3} read
$$\sum_{i\neq j} k_{1, {R_i}_x {R_j}_x} \, R_{i,\sigma} \, R_{j,\rho}\,(\lambda_j-\lambda_i) \= 0 \, ,$$
which implies
\begin{equation}\label{cross}
k_{1, {R_i}_x {R_j}_x} \= 0 \,,\qquad \forall~i\neq j \, .
\end{equation}

\begin{lem} Up to a total $x$-derivative, $k_1$ must have the form
\begin{equation} \label{k1}
k_1 \= \sum_{i=1}^n C_i(R_1,...,R_n){R_i}_x\log {R_i}_x \,-\, C_i(R_1,...,R_n){R_i}_x \+ \phi_i(R_1,...,R_n){R_i}_x
\end{equation}
for some $C_i,\phi_i$. Moreover, if $k_1$ has the form~\eqref{k1} then it satisfies \eqref{homo-1}, \eqref{homo-2}, \eqref{u3}.
\end{lem}
\begin{proof}
Equations \eqref{cross} imply that~$k_1$ must have the variable separation form
\begin{equation} \label{sep}
k_1 \= \sum_{i=1}^n B_i (R_1,...,R_n;{R_i}_x) \, .
\end{equation}
Noting that
\begin{align}
& k_{1,u^\alpha_x} \= \sum_{i=1}^n k_{1,{R_i}_x}R_{i,\alpha} \, , \nn \\
& k_{1,u^\alpha_x u^\beta_x} \= \sum_{i,j=1}^n k_{1,{R_i}_x{R_j}_x}R_{i,\alpha} R_{j,\beta} \, , \nn \\
& k_{1,u^\alpha_x u^\beta_x u^\gamma_x} \= \sum_{i,j,k=1}^n k_{1,{R_i}_x{R_j}_x{R_k}_x} R_{i,\alpha} R_{j,\beta}R_{k,\gamma} \nn
\end{align}
and substituting equation \eqref{sep} into equations \eqref{homo-1} we obtain
$$
{R_i}_x\,B_{i,{R_i}_x{R_i}_x{R_i}_x} \+ 2 B_{i,{R_i}_x{R_i}_x} \= 0 \, .
$$
If follows that
\begin{equation}\label{sol_ode}
B_i \= E_i(R) \+ \phi_i(R) {R_i}_x \+ C_i(R) {R_i}_x \log {R_i}_x  \,-\, C_i(R) {R_i}_x
\end{equation}
for some functions $C_i,\phi_i,E_i$. Finally, noticing that
\begin{align}
& k_{1,u^\beta} \= \sum_{i=1}^n \bigl(k_{1,{R_i}} R_{i,\beta}+k_{1,{R_i}_x} R_{i,\beta\sigma} u^\sigma_x\bigr) \, , \nn \\
& k_{1,u^\alpha_x u^\beta} \= \sum_{i,j=1}^n \left(k_{1,{R_i}_x R_j}R_{j,\beta}+k_{1,{R_i}_x {R_j}_x} R_{j,\beta\sigma}u^\sigma_x\right)R_{i,\alpha}+\sum_{i=1}^n k_{1,{R_i}_x}R_{i,\alpha\beta} \, , \nn \\
& k_{1,u^\alpha_x u^\beta_x u^\gamma} \= \sum_{i,j,k=1}^n \bigl(k_{1,{R_i}_x{R_j}_x R_k}R_{k,\gamma}+k_{1,{R_i}_x{R_j}_x {R_k}_x}R_{k,\gamma\sigma}u^\sigma_x\bigr)R_{i,\alpha} R_{j,\beta} \nn \\
& \qquad \qquad \qquad  \+ \sum_{i,j=1}^n k_{1,{R_i}_x{R_j}_x} \bigl(R_{i,\alpha\gamma} R_{j,\beta}+R_{i,\alpha} R_{j,\beta\gamma}\bigr) \, , \nn
\end{align}
and substituting \eqref{sep},\,\eqref{sol_ode} into \eqref{homo-2} we obtain
\beq
\p_\beta\biggl(\sum_{i=1}^n E_i(R)\biggr) \= 0 \, ,
\eeq
which finishes the proof.
\end{proof}

Now collect the terms of~\eqref{basic2} containing $u^\beta_{xx}u^\sigma_{xx}$:
\beq \label{qua}
u^\beta_{xx} \,  u^\sigma_{xx} \, \biggl(A^\alpha_\rho\frac{\p^3 k_1}{\p u^\alpha_x \p u^\beta_x\p u^\sigma_x}
+A^\alpha_\beta\frac{\p^3 k_1}{\p u^\alpha_x \p u^\rho_x\p u^\sigma_x}-2 A^\alpha_\sigma \frac{\p^3 k_1}{\p u^\alpha_x \p u^\beta_x \p u^\rho_x}\biggr) \= 0 \, .
\eeq
\begin{lem}
If $k_1$ satisfies~\eqref{cross} then it automatically satisfies~\eqref{qua}.
\end{lem}
\begin{proof} We have
\begin{align}
& \mbox{LHS of}~\eqref{qua}\nn\\
& \= u^\beta_{xx} \, u^\sigma_{xx}\sum_{i,j,l=1}^n k_{1,{R_i}_x{R_j}_x{R_l}_x}\, R_{l,\alpha}\left(A^\alpha_\rho R_{i,\beta}R_{j,\sigma}+A^\alpha_\beta R_{i,\rho} R_{j,\sigma}-2 A^\alpha_\sigma R_{i,\beta}R_{j,\rho}\right)\nn\\
& \= u^\beta_{xx} \, u^\sigma_{xx}\sum_{i=1}^n k_{1,{R_i}_x{R_i}_x{R_i}_x}\,R_{i,\alpha}\left(A^\alpha_\rho R_{i,\beta}R_{i,\sigma}+A^\alpha_\beta R_{i,\rho} R_{i,\sigma}-2 A^\alpha_\sigma R_{i,\beta}R_{i,\rho}\right)\nn\\
& \= u^\beta_{xx} \, u^\sigma_{xx}\sum_{i=1}^n k_{1,{R_i}_x{R_i}_x{R_i}_x}\,\lambda_i\left(R_{i,\rho} R_{i,\beta}R_{i,\sigma}+R_{i,\beta} R_{i,\rho} R_{i,\sigma}-2 R_{i,\sigma} R_{i,\beta}R_{i,\rho}\right)\nn\\
& \= 0 \,.\nn
\end{align}
The lemma is proved.
\end{proof}

Comparing the coefficients of~$u^\beta_{xx}$ of the both sides of \eqref{basic2} yields
\eqa
&\quad 2\,A^\alpha_\rho\, k_{1,u^\alpha_x u^\beta_x u^\gamma} u^\gamma_x \, - \, A^\alpha_\beta\, k_{1,u^\alpha_x u^\rho_x u^\gamma} u^\gamma_x \,-\, 3\,A^{\alpha}_{\beta\gamma} k_{1, u^\alpha_x u^\rho_x}u^\gamma_x \, - \, A^\alpha_{\gamma\epsilon}\,k_{1,u^\alpha_xu^\rho_x u^\beta_x}u^\epsilon_xu^\gamma_x\nn\\
& \+ A^\alpha_\beta\,\bigl(k_{1, u^\alpha_x u^\rho}-k_{1, u^\alpha u^\rho_x}\bigr) \+ A^\alpha_\rho \, \bigl(k_{1, u^\alpha_x u^\beta } -k_{1, u^\alpha u^\beta_x} \bigr) \,- \, 2\, \tilde d_{\rho\beta} \= 0 \, . \label{vxx}
\eeqa
Substituting \eqref{k1} into \eqref{vxx} we obtain the following lemma.
\begin{lem}
The functions $C_i$ must satisfy
\eqa \label{127}
C_{i,j} \= 0 \, ,\quad \forall~i\neq j.
\eeqa
\end{lem}
\begin{proof} Noting that
\begin{align}
& k_{1,{R_i}_x}=C_i \log {R_i}_x \+ \phi_i \,,  \nn \\
& k_{1,{R_i}_xR_j}=C_{i,j} \log {R_i}_x \+ \phi_{i,j} \,, \nn \\
& k_{1,{R_i}_x {R_j}_x}=C_i\, \delta_{ij}\,{R_i}_x^{-1} \,, \nn
\end{align}
we find that the only possible terms containing $\log {R_i}_x$ in equations \eqref{vxx} are
$$
A^\alpha_\rho\,\bigl( k_{1,u^\alpha_x u^\beta}-k_{1,u^\alpha u^\beta_x}\bigr) \,,\qquad 
A^\alpha_\beta\,\bigl( k_{1, u^\alpha_x u^\rho}- k_{1, u^\alpha u^\rho_x}\bigr) \, .
$$
If follows that
$\sum_{i,j=1}^n C_{i,j}(\lambda_i-\lambda_j) \bigl(R_{i,\beta}R_{j,\rho}+R_{i,\rho}R_{j,\beta}\bigr)\,\log {R_i}_x = 0$,
which yields 
\[\sum_{j\neq i}C_{i,j}(\lambda_i-\lambda_j)\bigl(R_{i,\beta}R_{j,\rho}+R_{i,\rho}R_{j,\beta}\bigr) \= 0\,.\]
This gives~\eqref{127}. The lemma is proved.
\end{proof}

\begin{lem}\label{dab}
The $\tilde d_{\alpha\beta}$ must have the form
\begin{equation}\label{dab-expression}
\tilde d_{\alpha\beta} \= - \, \frac{1}{2} \, \sum_{i=1}^n C_i(R_i)\left(\lambda_{i,\alpha}\, R_{i,\beta}+\lambda_{i,\beta}\, R_{i,\alpha}\right) \+ \frac{1}{2} \, \sum_{i\neq j} s_{ij} \,(\lambda_i-\lambda_j) \, R_{i,\alpha} R_{j,\beta} \, ,
\end{equation}
where $s_{ij}=\phi_{i,j}-\phi_{j,i}$ for some functions $\phi_i.$
\end{lem}
\begin{proof} 
Using equations \eqref{vxx} we obtain
\beq
2\,\tilde d_{\alpha\beta} \, u^\alpha_x \, u^\beta_x \= - \, 2 \, \sum_{i=1}^n C_i(R_i) \, {\lambda_i}_x \, {R_i}_x \+ \sum_{i,j=1}^n s_{ij}\,(\lambda_i-\lambda_j) \, {R_i}_x \, {R_j}_x \, .
\eeq
The lemma is proved.
\end{proof}

Finally, let us check that equalities~\eqref{basic2} hold true if $\tilde d_{\alpha\beta}$ and~$k_1$ are given by \eqref{dab-expression}~and~\eqref{k1}. Indeed, collecting the rest terms of both sides of \eqref{basic2} we find that it suffices to show
\begin{align}
& -  \bigl(\tilde d_{\alpha\beta,\rho} u^\beta_x u^\alpha_x -\,2\,\tilde d_{\rho\beta,\gamma}u^\gamma_x u^\beta_x\bigr)\nn\\
& ~ \= A^\alpha_\gamma u^\gamma_x \, \bigl(k_{1, u^\alpha u^\rho}-u^\sigma_x \, k_{1, u^\sigma u^\alpha u^\rho_x} \bigr)
\,-\, A^\alpha_\rho u^\gamma_x \, \bigl(k_{1,u^\gamma u^\alpha}-u^\sigma_x\, k_{1, u^\sigma u^\gamma u^\alpha_x}\bigr)\nn\\
& ~ \qquad   - \, A^{\alpha}_{\gamma\beta\epsilon} \, u^\epsilon_x u^\beta_x u^\gamma_x\, k_{1, u^\alpha_x u^\rho_x}
\+ A^\alpha_{\gamma\sigma} u^\sigma_x u^\gamma_x\bigl(k_{1, u^\alpha_x u^\rho}-k_{1, u^\alpha u^\rho_x}-u^\beta_x \,k_{1,u^\beta u^\alpha_x u^\rho_x}\bigr) \,,\label{final-check}
\end{align}
where $A^{\alpha}_{\gamma\beta\epsilon} \:= \eta^{\alpha\delta}\, \p_{\delta}\p_{\gamma}\p_{\beta}\p_{\epsilon}(h)$.
Indeed, the contribution of $\phi_i$-terms is just a result of canonical Miura-type transformation and note that equations \eqref{basic2} depend on $k_1$ \emph{linearly}, so we can assume $\phi_i=0$, $i=1,...,n$. Then by straightforward calculations we find that the both sides of the above equations \eqref{final-check} are equal to
$- \, \sum_{i=1}^n C_i(R_i) \bigl(\lambda_{i,\beta\delta}R_{i,\rho}+\lambda_{i,\rho}R_{i,\beta\delta}\bigr)\,u^\beta_x u^\delta_x$.

We have proved that the Hamiltonian perturbation~\eqref{pert-2} possesses quasi-triviality at $\mathcal{O}(\e^2)$-approximation 
iff $\tilde d_{\alpha\beta}$ has the form~\eqref{dab-expression}. 
We continue to show quasi-triviality implies integrability.

\begin{lem}\label{lemquasiimpliesint}
If a Hamiltonian perturbation of the form~\eqref{pert-2} with $H_1=0$ is 
quasi-trivial at the second order approximation, then it is $\mathcal{O}(\e^2)$-integrable.
\end{lem}
\begin{proof}
We have proved that there exist functions $C_i(R_i)$ and $\phi_i(R)$ such that equations \eqref{dab-expression} hold true. And the quasi-triviality is generated by $\e \, K_1+\mathcal{O}(\e^2):$
\beq
K_1 \= \int_{S^1}\sum_{i=1}^n C_i(R_i)\,{R_i}_x \log {R_i}_x-C_i(R_i){R_i}_x+\phi_i(R_1,...,R_n){R_i}_x \, dx.
\eeq
For a generic conservation law $F_0=\int_{S^1} f_0(v)\,dx$ of~\eqref{dispersionless}, denote by $\mu_1,...,\mu_n$ the distinct eigenvalues of the Hessian $(M^\alpha_\beta)$ of $f_0$. The calculations above can be applied to $F_0$, 
which give
\beq
F_2 \:= \{F_0,K_1\} \= \int_{S^1} \, \biggl(-\sum_{i=1}^n C_i(R_i)\, {\mu_i}_x{R_i}_x \+ \frac{1}{2}\,\sum_{i\neq j}\, (\mu_i-\mu_j)\, s_{ij}\, {R_i}_x {R_j}_x\biggr) \, dx \, .
\eeq
Then using the Jacobi identity we obtain
$\{H_0,F_2\} + \{H_2,F_0\} = 0$.
The lemma is proved.
\end{proof}

We proceed to prove that integrability implies quasi-triviality. 
According to Theorem~\ref{thm10}, it suffices to show that the expression~\eqref{sec-h2} and the expression~\eqref{dab-expression} are equivalent. 
Note that in the coordinate chart of the complete Riemann invariants, \eqref{dab-expression} becomes
\beq
d_{ij} \= -\frac{1}{2} \Bigl(C_i(R_i)\lambda_{i,j}+C_j(R_j) \lambda_{j,i}\Bigr) \+ \frac{1}{2} \, \sum_{i\neq j}^n s_{ij}\,(\lambda_i-\lambda_j)\,,
\eeq
where $s_{ij}=\phi_{i,j}-\phi_{j,i}$ for some functions $\phi_i.$
It then suffices to show that $-\frac{1}{2} \bigl(C_i(R_i)\lambda_{i,j}+C_j(R_j) \lambda_{j,i}\bigr)$, $\forall\,i\neq j$ can be absorbed into the term $\frac{1}{2}\sum_{i\neq j}^n s_{ij}\,(\lambda_i-\lambda_j).$ This is true because 
\begin{align}
&\p_{k}\biggl(\frac{C_i(R_i)\lambda_{i,j}+C_j(R_j) \lambda_{j,i}}{\lambda_i-\lambda_j}\biggr) \+ \p_{i}\biggl(\frac{C_j(R_j)\lambda_{j,k}+C_k(R_k) \lambda_{k,j}}{\lambda_j-\lambda_k}\biggr)\nn\\
&\qquad\qquad 
\+ \p_{j}\biggl(\frac{C_k(R_k)\lambda_{k,i}+C_i(R_i) \lambda_{i,k}}{\lambda_k-\lambda_i}\biggr) \= 0 \,,\quad \forall~\varepsilon(i,j,k) \= \pm1\, .
\end{align}
Part I is proved.  Part II then follows from Lemma~\ref{lemquasiimpliesint} and the above proved 
equivalence between~\eqref{sec-h2} and~\eqref{dab-expression}.
The theorem is proved. 
$\hfill\square$

\section{Example}\label{sect4}

The two component irrotational water wave equations in $1+1$ dimensions~\cite{AFM,Zakharov} are given by
\begin{align}
& \int_{-\infty}^\infty e^{-ikx} dx~ \biggl\{i\,
\eta_t \cosh{[k\,\epsilon\,(1+\mu\, \eta)]} - \frac{q_x}{\epsilon} \sinh{[ k\,\epsilon\,(1+\mu\, \eta)]}\biggr\} \= 0 \, , \label{Nlocal_a} \\
& q_t \+ \eta \+ \frac{\mu}{2}
q_x^2 \= \frac{\mu\epsilon^2}{2}\frac{(\eta+\mu\,
q_x\eta_x)^2}{1+\mu^2\epsilon^2\eta_x^2} \+ \frac{\sigma \epsilon^2\eta_{xx}}{(1+\mu^2\epsilon^2\eta_x^2)^{3/2}} \,.\label{Nlocal_b}
\end{align}
Here, $\mu$ and $\sigma$ are constants. For simplicity we will only consider the case $\sigma\equiv0$.
Denote $r=1+\mu\, \eta,~v=\mu\,q_x$. Then we can rewrite \eqref{Nlocal_a}--\eqref{Nlocal_b} 
as the perturbation of a system of Hamiltonian PDEs of hydrodynamic type:
\begin{align}
& r_t=(1+Q)^{-1}\sum_{j=1}^\infty
\f{(-1)^j\epsilon^{2j-2}}{(2j-1)!}\p_x^{2j-1}(r^{2j-1}v), \label{h_ww-1}\\
& v_t=-r_x-vv_x+\f{\epsilon^2}{2}\p_x\Biggl(\f{v\, r_x+(1+Q)^{-1}
\sum_{j=1}^\infty\f{(-1)^j\epsilon^{2j-2}}{(2j-1)!}\p_x^{2j-1}(r^{2j-1}v)}{1+\epsilon^2 r_x^2}\Biggr),\label{h_ww-2}
\end{align}
where $Q$ is an operator defined by $Q:=\sum_{j=1}^{\infty}\f{(-1)^j \epsilon^{2j}}{(2j)!}\p_x^{2j}\circ r^{2j}$.
The dispersionless limit of \eqref{h_ww-1}--\eqref{h_ww-2} was studied by Whitham~\cite{W} and is integrable.
Now we look at the second order approximation of \eqref{h_ww-1}--\eqref{h_ww-2}:
\begin{align}
& r_t \= - \, (rv)_x \+ \epsilon^2\Bigl(-r^2 r_x v_x-\f{1}{3} r^3v_{xx}\Bigr)_x \+ \mathcal{O}(\epsilon^4) \, , \label{h_ww_app-1}\\
& v_t \= - \, r_x \,-\, vv_x \+ \epsilon^2\Bigl(\f{1}{2} r^2v_x^2\Bigr)_x \+ \mathcal{O}(\epsilon^4) \, . \label{h_ww_app-2}
\end{align}
This approximation has the Hamiltonian structure: 
\begin{align}
& (r_t,v_t)^T\= \begin{pmatrix} 0 & \p_x \\ \p_x & 0 \\ \end{pmatrix} \biggl( \frac{\delta H}{\delta r(x)}, \,\frac{\delta H}{\delta v(x)} \biggr)^T,  \label{pab} \\
& H \= H_0 \+ \e^2\,H_2 \+ \mathcal{O}(\e^3), \label{pert-water wave} \\
& H_0 \= \, -\int_{S^1} \frac{1}{2} rv^2+\frac{r^2}{2}\,dx\, ,\qquad H_2 \= \int_{S^1} \frac{1}{6}r^3 v_x^2\,dx \, .
\end{align}

\begin{prp}\label{applicationwater}
The system \eqref{h_ww-1}--\eqref{h_ww-2} is not integrable in the sense of Definition~\ref{integrable}.
\end{prp}
\begin{proof} The Riemann invariants are 
$R_1=v/2+\sqrt{r}$, $R_2=v/2-\sqrt{r}$.
And the eigenvalues are
\beq
\lambda_1=-v-\sqrt{r}=-\frac{3}{2}R_1-\frac{1}{2} R_2,\quad \lambda_2=-v+\sqrt{r}=-\frac{1}{2}R_1-\frac{3}{2}R_2.
\eeq
This gives
$\lambda_{1,1}=\lambda_{2,2}=-3/2$.
According to Theorems \ref{thm10}--\ref{thm1} , the perturbation~\eqref{pert-water wave} is quasi-trivial 
at the second order approximation iff
the following equation has a solution:
\begin{align}
&- \bigl((R_1-R_2)(\phi_{1,2}-\phi_{2,1})\bigr){R_1}_x{R_2}_x\nn\\
&\quad\quad\+\frac{3}{2}\,C_1(R_1)\,{R_1}_x^2\,+\,\frac{3}{2}\,C_2(R_2)\,{R_2}_x^2 \,=\, \frac{(R_1-R_2)^6}{384}({R_1}_x+{R_2}_x)^2.\nn
\end{align}
However, the solution set to this equation is empty. 
The proposition is proved.
\end{proof}

Let us provide additional evidence supporting the already proved statement of Proposition~\ref{applicationwater}, 
which is more straightforward. It is easy to verify that 
system~\eqref{h_ww-1}--\eqref{h_ww-2} has 
four linearly independent conservation laws:
$$
\int_{S^1} r \,dx,~~\int_{S^1} v\,dx,~~\int_{S^1} rv\,dx,~~-H.
$$
We will show these form all the conservation laws for~\eqref{h_ww-1}--\eqref{h_ww-2}.
Indeed, 
only the following {\it four} conservation laws of the dispersionless limit of \eqref{h_ww-1}--\eqref{h_ww-2}
\beq\label{five}
\int_{S^1} r \,dx,~~\int_{S^1} v\,dx,~~\int_{S^1} r v\,dx,~~\int_{S^1} \frac{1}{2} r  v^2+\frac{r^2}{2}\,dx
\eeq
can be extended to conservation laws at the second order approximation for~\eqref{h_ww-1}--\eqref{h_ww-2}. 
To see this, denote $u^1=r$, $u^2=v$, and let
$$
F=F_0+\e^2 F_2+\mathcal{O}(\e^3)=\int_{S^1} f(u) \,dx+\e^2 \int_{S^1} D_{\alpha\beta}(u) \, u^\alpha_x u^\beta_x+\mathcal{O}(\e^3)
$$
be a conserved quantity of \eqref{h_ww-1}--\eqref{h_ww-2} at the second order approximation. Then we have
\begin{align}
& f_{vv}=r\, f_{rr}, \label{int-f}\\
& \mu_1=f_{rv}-\sqrt{r}\, f_{rr},\quad \mu_2=f_{rv}+\sqrt{r}\, f_{rr},\\
& d_{11}=d_{22}=\frac{1}{384}(r_1-r_2)^2,\\
& D_{11}=-\frac{\p_{R_1} (\mu_1)}{576}(r_1-r_2)^6,\quad D_{22}=-\frac{\p_{R_2} (\mu_2) }{576}(r_1-r_2)^6.
\end{align}
Substituting these equations in~\eqref{diij} and using~\eqref{int-f} we find
$f_{rrv}=0$.
It yields five solutions:
\beq
f=r,\quad f= v, \quad f=rv, \quad f=\frac12 r v^2+\frac12  r^2, \quad f=\frac{v^2}{2}+r\log r.
\eeq
However, through one by one verifications only the first four can be (and are indeed) extended to the second order approximation. 



\paragraph{Acknowledgements}
I would like to thank Youjin Zhang, Si-Qi Liu, and Boris Dubrovin for their advisings and helpful discussions. 
I am grateful to Youjin Zhang for suggesting the question about two-component water wave equations to me; I also thank 
Boris Dubrovin for suggesting the general question. 
Part of the work was done in Tsinghua University and SISSA; I thank Tsinghua University and SISSA for excellent working 
conditions and financial supports.

\end{document}